\documentclass[journal]{IEEEtran}
\usepackage{amsmath}
\usepackage{amsfonts}
\usepackage{amssymb}
\usepackage{algorithm}
\usepackage{algorithmic}
\usepackage{graphicx}
\usepackage{verbatim}

\newtheorem{theorem}{Theorem}
\newtheorem{lemma}{Lemma}
\newtheorem{remark}{Remark}

\newtheorem{corollary}{Corollary}
\newtheorem{example}{Example}

\newcommand{\beq}{\begin{equation}}
\newcommand{\eeq}{\end{equation}}
\newcommand{\beqnn}{\begin{equation*}}
\newcommand{\eeqnn}{\end{equation*}}
\newcommand{\beqy}{\begin{eqnarray}}
\newcommand{\eeqy}{\end{eqnarray}}
\newcommand{\beqynn}{\begin{eqnarray*}}
\newcommand{\eeqynn}{\end{eqnarray*}}
\newcommand{\bit}{\begin{itemize}}
\newcommand{\eit}{\end{itemize}}
\newcommand{\ben}{\begin{enumerate}}
\newcommand{\een}{\end{enumerate}}
\newcommand{\bex}{\begin{example}}
\newcommand{\eex}{\end{example}}

\newcommand{\balg}[1]{\begin{algorithm} \caption{#1}}
\newcommand{\ealg}{\end{algorithm}}

\newcommand{\balgc}{\begin{algorithmic}[1]}
\newcommand{\ealgc}{\end{algorithmic}}

\newcommand{\bary}{\begin{array}}
\newcommand{\eary}{\end{array}}
\newcommand{\bmx}{\begin{bmatrix}}
\newcommand{\emx}{\end{bmatrix}}
\newcommand{\bsmx}{\left[\begin{smallmatrix}}
\newcommand{\esmx}{\end{smallmatrix}\right]}
\newcommand{\bmxc}[1]{\left[\begin{array}{@{}#1@{}}}
\newcommand{\emxc}{\end{array}\right]}
\newcommand{\bcn}{\begin{center}}
\newcommand{\ecn}{\end{center}}

\newcommand{\A}{\boldsymbol{A}}
\newcommand{\B}{\boldsymbol{B}}

\newcommand{\D}{\boldsymbol{D}}

\newcommand{\I}{\boldsymbol{I}}

\renewcommand{\P}{\boldsymbol{P}}

\renewcommand{\S}{\boldsymbol{S}}

\newcommand{\U}{\boldsymbol{U}}

\newcommand{\e}{\boldsymbol{e}}

\newcommand{\rr}{\boldsymbol{r}}

\renewcommand{\u}{\boldsymbol{u}}
\renewcommand{\v}{\boldsymbol{v}}
\newcommand{\w}{\boldsymbol{w}}
\newcommand{\x}{{\boldsymbol{x}}}
\newcommand{\y}{{\boldsymbol{y}}}

\newcommand{\0}{{\boldsymbol{0}}}

 \usepackage{cite}

\begin{document}
%
\title{A Novel Sufficient Condition for  Generalized Orthogonal Matching Pursuit}

\author{Jinming~Wen, Zhengchun Zhou, Dongfang Li and  Xiaohu Tang
\thanks{This research  was supported by  ``Programme Avenir
Lyon Saint-Etienne de l'Universit\'e de Lyon" in the framework of the programme
``Inverstissements d'Avenir" (ANR-11-IDEX-0007),  ANR through the HPAC project under Grant ANR~11~BS02~013, and  NSFC (Nos. 61661146003, 61672028).}
\thanks{J.~Wen is with  Department of Electrical and Computer Engineering, University of Alberta, Edmonton T6G 2V4 (e-mail: jwen@math.mcgill.ca).}
\thanks{Z. Zhou is with the School of Mathematics, Southwest Jiaotong University,
Chengdu 610031, China (e-mail: zzc@home.swjtu.edu.cn).}
\thanks{D. Li is with the School of Mathematics and Statistics,
Huazhong University of Science and Technology, Wuhan 430074, China (e-mail: dfli@hust.edu.cn).}
\thanks{X. Tang is with the Information Security and National Computing Grid
Laboratory, Southwest Jiaotong University, Chengdu 610031, China (e-mail: xhutang@swjtu.edu.cn).}

}

\maketitle

\begin{abstract}
Generalized orthogonal matching pursuit (gOMP),
also called orthogonal multi-matching pursuit,
is an extension of OMP in the sense that $N\geq1$ indices are identified per iteration.
In this paper, we show that if the restricted isometry constant (RIC) $\delta_{NK+1}$ of a sensing matrix $\A$ satisfies $\delta_{NK+1} < 1/\sqrt {K/N+1}$,
then under a condition on the signal-to-noise ratio,
gOMP identifies at least one index in the support of any $K$-sparse signal $\x$
from $\y=\A\x+\v$ at each iteration, where $\v$ is a noise vector.
Surprisingly, this condition does not require $N\leq K$ which is needed in Wang, \textit{et al} 2012 and Liu, \textit{et al} 2012. Thus, $N$ can have more choices.
When $N=1$, it reduces to be a sufficient condition for OMP,
which is less restrictive than that proposed in Wang 2015.
Moreover, in the noise-free case, it is a sufficient condition
for accurately recovering $\x$ in $K$ iterations which is less restrictive than the best known one.
In particular, it reduces to the sharp condition proposed in Mo 2015 when $N=1$.
\end{abstract}



\section{Introduction}\label{introduction}
One of the central aims of compressed sensing  is to recover a $K$-sparse unknown signal
$\x\in \mathbb{R}^n$ (i.e., $\x$ has at most $K$ nonzero entries)
from the following linear model \cite{CanT05} \cite{Don06}
\beq
\label{e:model}
\y=\A\x+\v,
\eeq
where $\y\in \mathbb{R}^m$ is an observation vector,
$\A\in \mathbb{R}^{m\times n}$ (with $m<<n$) is a given sensing matrix and $\v \in \mathbb{R}^{m}$ is a noise vector.

It has been shown that (see, e.g., \cite{CanT05, Don06, TroG07, WenLZ15}) stably recovering $\x$
by some sparse recovery algorithms is possible under certain conditions on $\A$.
One of the widely used frameworks for characterizing such conditions is the restricted isometry property (RIP) \cite{CanT05}.
For a sensing matrix $\A$ and for any integer $K$, the restricted isometry constant (RIC) $\delta_K$ of order $K$ is defined as the smallest constant such that
\begin{equation}
\label{e:RIP}
(1-\delta_K)\|\x\|_2^2\leq \|\A\x\|_2^2\leq(1+\delta_K)\|\x\|_2^2
\end{equation}
for all $K$-sparse vectors $\x$.

One of the most popular sparse recovery algorithms is orthogonal matching pursuit (OMP) \cite{TroG07}.
Generalized orthogonal matching pursuit (gOMP) \cite{WanKS12},
also called orthogonal multi-matching pursuit \cite{LiuT12},
is an extension of OMP in the sense that $N (N\geq1)$ indices are identified per iteration.
Simulations in \cite{WanKS12} and \cite{LiuT12} indicate  that, compared with OMP,
gOMP has better sparse recovery performance.
The gOMP algorithm is described in Algorithm \ref{a:gOMP}, where
$\A_S$ denotes the submatrix of $\A$ that contains only the columns indexed by set $S\subset\{1,2,\ldots ,n\}$,
$\x_S$ denotes the subvector of $\x$ that contains only the entries indexed by $S$.
Note that when $N=1$, gOMP reduces to OMP.

\begin{algorithm}[h!]
\caption{gOMP}  \label{a:gOMP}
Input: $\y\in \mathbb{R}^m$, $\A\in \mathbb{R}^{m\times n}$, $K$,
$N\leq (m-1)/K$ and $\epsilon>0$ .\\
Initialize: $k=0, \rr^0=\y, S_0=\emptyset$.

\begin{algorithmic}[1]
\WHILE{$k<K$ and $\|\rr^k\|_2>\epsilon$}
\STATE $k=k+1$
\STATE Choose indexes $i_1,\ldots, i_N$ corresponding to the $N$ largest magnitude of $\A^T\rr^{k-1}$,
\STATE $S_{k}=S_{k-1}\bigcup\{i_1,\ldots, i_N\}$,
\STATE $\hat{\x}_{S_{k}}=\arg \min\limits_{\x\in \mathbb{R}^{|S_{k}|}}\|\y-\A_{S_{k}}\x\|_2$,
\STATE $\rr^{k}=\y-\A_{S_{k}}\hat{\x}_{S_{k}}$
\ENDWHILE
\end{algorithmic}
Output: $\hat{\x}=\arg \min\limits_{\x: \Omega=S_k}\|\y-\A\x\|_2$.
\end{algorithm}




Many RIC-based conditions have been proposed to guarantee the accurately recovery of $K$-sparse signals with gOMP
in the noise-free case (i.e., when $v=0$) for general $N$,
such as $\delta_{NK} < 1/(\sqrt {K/N}+3) $ \cite{WanKS12},
$\delta_{NK} < 1/\big((2+\sqrt{2})\sqrt {K/N}\big )$ \cite{LiuT12},
$\delta_{NK} < 1/(\sqrt {K/N}+2)$ and $\delta_{NK+1} < 1/(\sqrt {K/N}+1)$  \cite{SatDC13}.
Recently, it was further improved to $\delta_{NK} < 1/(\sqrt {K/N}+1.27)$ \cite{SheLPL14}.
It is worthwhile pointing out that there are more sufficient conditions for OMP,
see, e.g., \cite{DavW10,WenZL13,Mo15}.

Sufficient conditions of the exact support recovery of $K$-sparse signals with gOMP in the noisy case  have also
been widely studied (see e.g., \cite{LiSRK15} \cite{LiSWL15}).
In particular, it was proved in \cite{LiSWL15} that under certain conditions on the minimum magnitude of the nonzero elements of $\x$,
$\delta_{NK+1} < 1/(\sqrt {K/N}+1)$ is a sufficient condition under both
$\ell_2$ and $\ell_{\infty}$ bounded noises
(i.e., $\|\v\|_2\leq \epsilon$ and $\|\A^T\v\|_{\infty}\leq \epsilon$ for some constant $\epsilon$, respectively) .


%


In this paper, we aim to investigate RIP based sufficient conditions for
the exact support recovery with gOMP in the noisy case.
Instead of considering the $\ell_2$ and $\ell_{\infty}$ bounded noises separately (see, e.g, \cite{LiSWL15}),
we follow \cite{Wan15} and use the
signal-to-noise ratio (SNR) and the minimum-to-average ratio (MAR), which are respectively defined by
\begin{align}
\label{e:SNR}
\mbox{SNR}&=
\begin{cases}
\frac{\|\A\x\|^2_2}{\|\v\|_2^2} & \v\neq\0 \\
+\infty & \v=\0
\end{cases}
\mbox{ and }
\mbox{MAR}&=\frac{\min_{i\in \Omega}|x_i|^2}{\|\x\|_2^2/K},
\end{align}
to measure $\v$ and $\x$.
The main reason that we use SNR is because it is a commonly used  measure
that compares the level of a desired signal to the level of background noise in science and engineering.
We show that under a condition on SNR and MAR,
gOMP is ensured to recover at least one index in the support of  $\x$ at each iteration if $\delta_{NK+1} < 1/\sqrt {K/N+1}$.
As consequences, we have:
\begin{itemize}
\item
Unlike \cite{WanKS12}  and \cite{LiuT12}, which require $N\leq \min(K,m/K)$,
our condition on $N$ is only $N\leq (m-1)/K$ which ensures that the assumption $\delta_{NK+1} < 1/\sqrt {K/N+1}$ makes sense.
This allows more choices of $N$ for gOMP.

\item
The exact support recovery condition for  gOMP reduces to that for OMP when $N=1$,
and it is weaker than that proposed in \cite{Wan15} in terms of both SNR and RIP.

\item
In the noise-free case, we obtain that $\delta_{NK+1} < 1/\sqrt {K/N+1}$ is a sufficient condition
for accurately recovering $K$-sparse signals with gOMP in $K$ iterations.
This improves the best known  condition $\delta_{NK+1} < 1/(\sqrt {K/N}+1)$ \cite{SatDC13}.
Moreover,  when $N=1$, it is a sharp condition according to  \cite{WenZL13} \cite{Mo15}.

%
\end{itemize}

The rest of the paper is organized as follows.
We give some useful notation and lemmas in section \ref{s:lemmas}.
We present our main results in Section \ref{s:main}, and do numerical tests to illustrate them in
Section \ref{s:num}.
Finally, this paper is summarized in Section \ref{s:con}.

\section{Notation and  Useful Lemmas}
\label{s:lemmas}

We introduce some notations and useful lemmas in this section.
\subsection{Notation}
Throughout this paper, we adopt the following notation unless otherwise stated.
Let $\mathbb{R}$ be the real field. Boldface lowercase letters denote column vectors, and boldface uppercase letters denote matrices.
e.g., $\x\in\mathbb{R}^n$ and $\A\in\mathbb{R}^{m\times n}$.
Let $\0$ denote a zero vector.
Let $\Omega$ be the support of $\x$ and $|\Omega|$ be the cardinality of $\Omega$.
Let set $S\subset \{1,2,\ldots ,n\}$, and $\Omega \setminus S=\{i|i\in\Omega, i\not\in S\}$.
Let $\Omega^c$ and $S^c$ be the complement of $\Omega$ and $S$,
i.e., $\Omega^c=\{1,2,\ldots ,n\}\setminus \Omega$, and $S^c=\{1,2,\ldots ,n\}\setminus S$.
Let $\A_S$ be the submatrix of $\A$ that only contains the columns indexed by $\S$,
and $\x_S$ be the subvector of $\x$ that only contains the entries indexed by $\S$,
and $\A_S^T$ be the transpose of $\A_S$.
For any full column rank matrix $\A_S$, let $\P_S=\A_S(\A_S^T\A_S)^{-1}\A_S^T$
and $\P^{\bot}_S=\I-\P_S$ denote the projector and the orthogonal complement projector
on the column space of $\A_S$, respectively.

\subsection{Useful lemmas}
We now introduce some lemmas that will be used in the sequel.
\begin{lemma}[\cite{CanT05}]
\label{l:monot}
If a matrix $\A\in\mathbb{R}^{m\times n}$  satisfies the RIP of orders $K_1$ and $K_2$ with $K_1<K_2$, then
$
\delta_{K_1}\leq \delta_{K_2}.
$
\end{lemma}

\begin{lemma}[\cite{SheLPL14}]
\label{l:orthogonalcomp}
Let $S_1,S_2$ be two subsets of  $\{1,2,\ldots ,n\}$ with $|S_2\setminus S_1|\geq1$. If a matrix $\A\in\mathbb{R}^{m\times n}$  satisfies the RIP of order $|S_1\cup S_2|$, then
for any vector $\x \in \mathbb{R}^{|S_2\setminus S_1|}$,
\beqnn
(1-\delta_{|S_1\cup S_2|})\|\x\|_2^2\leq
\|\P^{\bot}_{S_1}\A_{S_2\setminus S_1}\x\|_2^2\leq(1+\delta_{|S_1\cup S_2|})\|\x\|_2^2.
\eeqnn
\end{lemma}

\begin{lemma}[\cite{NeeT09}]
\label{l:AtRIP}
Let $\A$ satisfy the RIP of order $K$ and $S$ be a subset of  $\{1,2,\ldots ,n\}$  with $|S|\leq K$, then for any $\x \in \mathbb{R}^m$,
$
\|\A^T_S\x\|_2^2\leq(1+\delta_K)\|\x\|_2^2.
$
\end{lemma}

\section{Main results}
\label{s:main}

We propose our main results in this section. We begin with the following technical lemma.

\begin{lemma}
\label{l:main}
Let set $S\subseteq \{1,2,\ldots,n\}$ satisfy $|S|=kN$ and $|\Omega\cap S|=\ell$
for some integers $N$, $k$ and $\ell$ with $0\leq k\leq \ell\leq |\Omega|-1$ and $N(k+1)+|\Omega|-k\leq m$.
Let $W\subseteq \Omega^c$ satisfy $|W|=N$ and $W\cap S=\emptyset$.
If $\A$ in \eqref{e:model} satisfies the RIP of order $N(k+1)+|\Omega|-\ell$,
then
\begin{align}
\label{e:main}
&\max_{i\in \Omega\setminus S}|\A_{i}^T\P^{\bot}_S\A_{\Omega\setminus S}\x_{\Omega\setminus S}|
-\frac{1}{N}\sum_{j\in W}|\A_{j}^T\P^{\bot}_S\A_{\Omega\setminus S}\x_{\Omega\setminus S}| \nonumber \\
\geq &\frac{(1-\sqrt {(|\Omega|-\ell)/N+1 }\delta_{N(k+1)+|\Omega|-\ell})\|\x_{\Omega\setminus S}\|_2}{\sqrt{|\Omega|-\ell}}.
\end{align}
\end{lemma}


Note that Lemma \ref{l:main}  extends \cite[Lemma 1]{WenZWTM15} for $N=1$ to general $N$,
and will play a key role in proving Theorem \ref{t:kstep} below.
Although it is motivated by \cite[Lemma 1]{WenZWTM15} and~\cite[Lemma II.2]{Mo15},
it is stronger than \cite[Lemma 1]{WenZWTM15} and~\cite[Lemma II.2]{Mo15}
since  it holds for general $N$ and for the noisy case (which contains the noise-free case as a special case).
In contrast, \cite[Lemma 1]{WenZWTM15} is useful only when $N=1$,
and \cite[Lemma II.2]{Mo15} is applicable only when $N=1$ and $\v=\0$.
In addition, regarding the proof itself, there are two key distinctions between Lemma \ref{l:main} and \cite[Lemma 1]{WenZWTM15}.
Due to the limitation of space, the proof of Lemma \ref{l:main},
the connections and differences between it and that of \cite[Lemma 1]{WenZWTM15} are detailed in the supplementary file.

\begin{remark}
The condition $N(k+1)+|\Omega|-k\leq m$ in Lemma \ref{l:main} is to ensure the assumption
that $\A$ satisfies the RIP of order $N(k+1)+|\Omega|-\ell$ makes sense.
\end{remark}



With Lemma \ref{l:main}, we can prove the following theorem.
\begin{theorem}
\label{t:kstep}
Let $\A$ satisfy the RIP with
\beq
\label{e:delta}
\delta_{N(k+1)+|\Omega|-k} < \frac{1}{\sqrt {|\Omega|/N+1 }}
\eeq
for some integers $k$ and $N$ satisfying $0\leq k\leq |\Omega|-1$ and $N(k+1)+|\Omega|-k\leq m$.
Then gOMP identifies at least one index in $\Omega$ in each of the first $k+1$ iterations
until all the indexes in $\Omega$ are selected or gOMP terminates provided that
\begin{align}
\label{e:SNRld}
\sqrt{\mbox{SNR}}>\frac{\sqrt{2K}(1+\delta_{N(k+1)+|\Omega|-k})}{(1-\sqrt {|\Omega|/N+1 }\delta_{N(k+1)+|\Omega|-k})\sqrt{\mbox{MAR}}}.
\end{align}
\end{theorem}

\begin{proof}
See Appendix \ref{ss:kstep}.
\end{proof}


By Theorem \ref{t:kstep} with $k=|\Omega|-1$ and Lemma \ref{l:monot}, we can obtain Theorem \ref{t:main1} below.

\begin{theorem}
\label{t:main1}
Let $\A$  satisfy the RIP with
\beq
\label{e:delta2}
\delta_{NK+1} < \frac{1}{\sqrt {K/N+1 }},
\eeq
for an integer $N$ with $1\leq N\leq (m-1)/K$.
Then  gOMP either identifies at least $k_0$ indexes in $\Omega$ if gOMP terminates after performing $k_0$ iterations with $1\leq k_0< K$
or  recovers $\Omega$ in $K$ iterations provided that
\begin{align}
\label{e:SNRld2}
\sqrt{\mbox{SNR}}>\frac{\sqrt{2K}(1+\delta_{NK+1})}{(1-\sqrt {|\Omega|/N+1 }\delta_{NK+1})\sqrt{\mbox{MAR}}}.
\end{align}
\end{theorem}

When $N=1$, gOMP reduces to OMP, and the following result can be directly obtained from Theorem \ref{t:main1}.
\begin{corollary}
\label{c:main1}
Let $\A$ satisfy the RIP with
$\delta_{K+1} < 1/\sqrt {K+1}$.
Then  OMP either identifies at least $k_0$ indexes in $\Omega$ if it terminates after performing $k_0$ iterations with $1\leq k_0< K$
or it recovers $\Omega$ in $K$ iterations provided that
\begin{align}
\label{e:SNRldN=1}
\sqrt{\text{SNR}} > \frac{\sqrt{2K}(1+\delta_{K+1})}{(1-\sqrt {|\Omega|+1 }\delta_{K+1})\sqrt{\mbox{MAR}}}.
\end{align}
\end{corollary}

\begin{remark}
The recovery condition for OMP in \cite[Theorem 3.1]{Wan15} is
\[
\delta_{K+1} < \frac{1}{\sqrt {K }+1},\,\;
\sqrt{\mbox{SNR}} > \frac{2 \sqrt{K} (1 + \delta_{K + 1})}{(1 - (\sqrt K + 1)\delta_{K + 1}) \sqrt{\text{MAR}}}.
\]
Clearly, our sufficient condition given  by Corollary \ref{c:main1} is less restrictive
than that given by \cite[Theorem 3.1]{Wan15} in terms of both RIC and SNR.

\end{remark}

Notice that gOMP  may terminate after performing $k_0$ with $0<k_0<K$ iterations,
and in this case $\Omega$ is not guaranteed to be recovered by gOMP under \eqref{e:delta2} and \eqref{e:SNRld2}.
However, we have:
\begin{theorem}
\label{t:main2}
Suppose that $\v=\0$, 
and $\A$ satisfies the RIP with \eqref{e:delta2}
for an integer $N$ with $1\leq N\leq (m-1)/K$. Then gOMP recovers $\x$ in $K$ iterations.
\end{theorem}

\begin{remark}
In the noise-free case, the best known  condition on $\delta_{NK+1}$
for accurately recovering $\x$ with gOMP in $K$ iterations
is $\delta_{NK+1} < 1/(\sqrt {K/N}+1)$ \cite{SatDC13}.
Obviously,  our sufficient condition given by Theorem \ref{t:main2} is less restrictive.
\end{remark}

Note that Theorem \ref{t:main2} can be directly obtained from Theorem \ref{t:main1} and Lemma \ref{l:stop} below.

\begin{lemma}
\label{l:stop}
Suppose that  $\v=\0$, 
and $\A$ satisfies the RIP with \eqref{e:delta}
for some integers $k$ and $N$ with $1\leq k\leq |\Omega|-1$ and  $1\leq N\leq (m-1)/K$.
If  there exists an integer  $k_0$ with $0<k_0\leq k$
and $|\Omega\cap  S_{k_0}|\geq k_0$ such that $||r^{k_0}||_2=0$
(see Algorithm  \ref{a:gOMP} for the definitions of $S_{k_0}$ and $r^{k_0}$ ).
Then $\Omega\subseteq S_{k_0}$.
\end{lemma}

\begin{proof}
We prove this lemma  by contradiction.
Suppose that $\Omega\not\subseteq S_{k_0}$ and let $\Gamma=\Omega\cup S_{k_0}$.
Let $\bar{\x}, \tilde{\x}\in \mathbb{R}^{|\Gamma|}$ satisfy
$\bar{x}_i=x_i$ for $i\in\Omega $ and $\bar{x}_i=0$ for $i\notin\Omega $,
and $\tilde{x}_i=(\hat{\x}_{S_{k_0}})_i$ for $i\in S_{k_0} $ and $\tilde{x}_i=0$ for $i\notin S_{k_0}$, where $\hat{\x}_{S_{k_0}}$ is
the vector generated by Algorithm \ref{a:gOMP}.
Since $\|r^{k_0}\|_2=0$, by line 6 of Algorithm \ref{a:gOMP}, $\A_{S_{k_0}}\hat{\x}_{S_{k_0}}=\y$, we have
\beq
\label{e:xbarxtilde}
\A_{\Gamma}\bar{\x}=\A_{\Omega}\x_{\Omega}=\A\x=\y=\A_{S_{k_0}}\hat{\x}_{S_{k_0}}=\A_{\Gamma}\tilde{\x}.
\eeq
Note that $|\Omega\cap  S_{k_0}|\geq k_0$  and $\Gamma=\Omega\cup S_{k_0}$. Thus
\[
|\Gamma|=|\Omega|+|S_{k_0}|-|\Omega\cap  S_{k_0}|
\leq |\Omega|+Nk-k\leq N(k+1)+|\Omega|-k.
\]
By \eqref{e:delta}, $\A_{\Gamma}$ is full column rank. Thus, applying \eqref{e:xbarxtilde} yields $\bar{\x}=\tilde{\x}$.

On the other hand, by the definitions of $\bar{\x}$ and $\tilde{\x}$, and the assumption that $\Omega\not\subseteq S_{k_0}$,
there exists $j\in (\Omega \setminus S_{k_0})$ such that $\bar{x}_j\neq 0$ but  $\tilde{x}_j=0$. This implies that
$\bar{\x}\neq\tilde{\x}$ which contradicts with $\bar{\x}=\tilde{\x}$. Completing the proof.
\end{proof}

\begin{remark}
When $N=1$, Theorem \ref{t:main2} reduces to \cite[Theorem III.1]{Mo15}.
\end{remark}

\section{Numerical tests}
\label{s:num}
In this section, we do numerical tests to illustrate Theorems \ref{t:main1} and \ref{t:main2}.
Since constructing general non-square deterministic matrices satisfying RIP with a given
RIC is still an open problem, we use square sensing matrices to do tests.
Specifically, for each given $K$ and $N$, we assume $n=NK+1$ and let $\A=\D\U$,
where $\D\in \mathbb{R}^{n\times n}$ is a diagonal matrix with $d_{ii}$ being uniformly distributed
over $\left[\sqrt{1-\frac{0.99}{\sqrt {K/N+1 }}},\sqrt{1+\frac{0.99}{\sqrt {K/N+1 }}}\right]$
for $1\leq i \leq n$, and $\U\in \mathbb{R}^{n\times n}$ is an orthogonal matrix
obtained by the QR factorization of a random matrix whose entries
independent and identically follow the standard normal distribution.
Then, by the definition of RIP, one can easily verify that $\A$ satisfies the RIP with \eqref{e:delta2}.
For a given $K$, we generate a $K$-sparse vector $\x\in \mathbb{R}^{n}$.
To illustrate Theorems \ref{t:main1} and \ref{t:main2}, we respectively assume
$\v=\frac{\|\A\x\|_2}{\sqrt{\mbox{SNR}}}\frac{\bar{\v}}{\|\bar{\v}\|_2}$ and $\v=\0$,
where $\bar{\v} \sim \mathcal{N}(\0,\I)$, and
\[
\sqrt{\mbox{SNR}}=0.01+\frac{\sqrt{2K}(1+\delta_{NK+1})}{(1-\sqrt {|\Omega|/N+1 }\delta_{NK+1})\sqrt{\mbox{MAR}}}.
\]
Note that $\mbox{MAR}$ can be computed via \eqref{e:SNR} and
$\delta_{NK+1}=\max\{1-\min_{1\leq i\leq n}d_{ii}, \max_{1\leq i\leq n}d_{ii}-1\}$.
Clearly, \eqref{e:SNRld2} holds.
After generating $\A,\x$ and $\v$, $\y$ can be computed via \eqref{e:model}.
Finally, we set $\epsilon=\|\v\|_2$ and use gOMP to recover $\x$.
We did lots of tests by choosing different $K$ and $N$ and found that
gOMP can always accurately recovering $\x$ in the noise-free case and
find its support in the noisy case.

\section{Conclusion}
\label{s:con}

In this paper, we have shown that under some conditions on SNR and MAR, $\delta_{NK+1} < 1/\sqrt {K/N+1}$
is a sufficient condition for the exact support recovery of $K$-sparse signals with gOMP.
Surprisingly, unlike that in \cite{WanKS12} and \cite{LiuT12}, this condition does not require $N\leq K$ which provides more choices for $N$.
When $N=1$, it is a sufficient condition for OMP and it is  better than that proposed in \cite{Wan15}.
In the noise-free case, it is a sufficient condition for accurately recovering $K$-sparse signals
with gOMP in $K$ iterations,
which is  better than the best known one in terms of $\delta_{NK+1}$ in \cite{SatDC13}.
Moreover, it reduces to the sharp condition in \cite{Mo15} when $N=1$.



\begin{appendices}
\section{Proof of Lemma \ref{l:main}}

In the following, we extend the proof of \cite[Lemma 1]{WenZWTM15} 
for $N=1$ to general $N$ to prove Lemma \ref{l:main}.
Although our proof is highly relying on the techniques used in proving
\cite[Lemma 1]{WenZWTM15}  and \cite[Lemma II.1]{Mo15}, 
there are two main distinctions between these proofs.
On the one hand, instead of defining a scalar $t$ as in 
\cite[Lemma 1]{WenZWTM15}  and \cite[Lemma II.1]{Mo15},
we need to define a vector $\overline{\e}_{W}\in \mathbb{R}^{N}$ (see \eqref{e:ew}) to explore the fact that $|W|=N$.
On the other hand, the choice of $\alpha$ (see \eqref{e:alpha}) is also  different.
One can see from the following proof that both the well-defined $\overline{\e}_{W}$ and
well-chosen $\alpha$ play a key role in proving Lemma \ref{l:main}.

{\em Proof of Lemma \ref{l:main}.}
By \cite[(21)]{WenZWTM15}, we have
\begin{align}
\label{e:Ax}
&\|\P^{\bot}_S\A_{\Omega\setminus S}\x_{\Omega\setminus S}\|_2^2\nonumber\\
\leq&\sqrt{|\Omega|-\ell}\|\x_{\Omega\setminus S}\|_2\max_{i\in \Omega\setminus S}|\A_{i}^T\P^{\bot}_S\A_{\Omega\setminus S}\x_{\Omega\setminus S}|.
\end{align}
In fact, since $|\Omega\cap S|=\ell\leq |\Omega|-1$, $\|\x_{\Omega\setminus S}\|_1\neq0$. Thus, we obtain
\begin{align*}
\,\;&\max_{i\in \Omega\setminus S}|\A_{i}^T\P^{\bot}_S\A_{\Omega\setminus S}\x_{\Omega\setminus S}|\\
=&\frac{1}{\|\x_{\Omega\setminus S}\|_1}
(\sum_{j\in \Omega\setminus S}|x_j|)\max_{i\in \Omega\setminus S}|\A_{i}^T\P^{\bot}_S\A_{\Omega\setminus S}\x_{\Omega\setminus S}|\\
\overset{(a)}{\geq}&\frac{1}{\sqrt{|\Omega|-\ell}\|\x_{\Omega\setminus S}\|_2}(\sum_{j\in \Omega\setminus S}|x_j|)
\max_{i\in \Omega\setminus S}|\A_{i}^T\P^{\bot}_S\A_{\Omega\setminus S}\x_{\Omega\setminus S}|\\
\geq&\frac{1}{\sqrt{|\Omega|-\ell}\|\x_{\Omega\setminus S}\|_2}\sum_{j\in \Omega\setminus S}
\big(|x_j\A_j^T\P^{\bot}_S\A_{\Omega\setminus S}\x_{\Omega\setminus S}|\big)\\
\geq&\frac{1}{\sqrt{|\Omega|-\ell}\|\x_{\Omega\setminus S}\|_2}\sum_{j\in \Omega\setminus S}
\big(x_j\A_j^T\P^{\bot}_S\A_{\Omega\setminus S}\x_{\Omega\setminus S}\big)\\
=&\frac{1}{\sqrt{|\Omega|-\ell}\|\x_{\Omega\setminus S}\|_2}
\big(\sum_{j\in \Omega\setminus S}x_{j}\A_{j}\big)^T\P^{\bot}_S\A_{\Omega\setminus S}\x_{\Omega\setminus S}\\
=&\frac{1}{\sqrt{|\Omega|-\ell}\|\x_{\Omega\setminus S}\|_2}
\big(\A_{\Omega\setminus S}\x_{\Omega\setminus S}\big)^T\P^{\bot}_S\A_{\Omega\setminus S}\x_{\Omega\setminus S}\\
\overset{(b)}{=}&\frac{1}{\sqrt{|\Omega|-\ell}\|\x_{\Omega\setminus S}\|_2}\|\P^{\bot}_S\A_{\Omega\setminus S}\x_{\Omega\setminus S}\|_2^2,
\end{align*}
where (a) follows from $|\text{supp}(\x_{\Omega\setminus S})|=|\Omega|-\ell$ and the Cauchy-Schwarz inequality,
and (b) is from
\begin{align}
\label{e:orthcom}
(\P^{\bot}_S)^T\P^{\bot}_S=\P^{\bot}_S\P^{\bot}_S=\P^{\bot}_S.
\end{align}
Thus, \eqref{e:Ax} holds.

Let
\beq
\label{e:alpha}
\alpha=-\frac{\sqrt{(|\Omega|-\ell)/N+1}-1}{\sqrt{(|\Omega|-\ell)/N}},
\eeq
then by one can easily verify that
\beq
\label{e:alphaproperty}
\frac{2\alpha}{1-\alpha^2}=-\sqrt{\frac{|\Omega|-\ell}{N}}, \quad
\frac{1+\alpha^2}{1-\alpha^2}=\sqrt{\frac{|\Omega|-\ell}{N}+1}.
\eeq

To simplify notation, let $W=\{j_1, j_2,\ldots, j_N\}$
and define $\overline{\e}_{W}\in \mathbb{R}^{N}$ with
\beq
\label{e:ew}
(\overline{e}_{W})_i=
 \begin{cases}
      1 & \mbox{if }\A^T_{j_i}\P^{\bot}_S\A_{\Omega\setminus S}\x_{\Omega\setminus S}\geq 0 \\
      -1 & \mbox{if }\A^T_{j_i}\P^{\bot}_S\A_{\Omega\setminus S}\x_{\Omega\setminus S}< 0
   \end{cases}
, \quad 1\leq i\leq N.
\eeq
Then,
\beq
\label{e:sum}
\overline{\e}_{W}^T\A^T_{W}\P^{\bot}_S\A_{\Omega\setminus S}\x_{\Omega\setminus S}
=\sum_{j\in W}|\A_{j}^T\P^{\bot}_S\A_{\Omega\setminus S}\x_{\Omega\setminus S}|.
\eeq

Furthermore, define
\begin{align}
\label{e:B}
\B=&\P^{\bot}_S
\bmx
\A_{\Omega\setminus S}&\A_{ W}
\emx,\\
\u=&
\bmx
\x_{\Omega\setminus S}\\ \0
\emx\in \mathbb{R}^{|\Omega\setminus S|+N}, \nonumber\\
\w=&
\bmx
\0\\ \frac{\alpha\|\x_{\Omega\setminus S}\|_2}{\sqrt{N}}\overline{\e}_{W}
\emx\in \mathbb{R}^{|\Omega\setminus S|+N}.\nonumber
\end{align}
Then,
\begin{align}
\label{e:AB}
\B\u=\P^{\bot}_S\A_{\Omega\setminus S}\x_{\Omega\setminus S},
\end{align}
and
\begin{align}
\label{e:uw}
\|\u+\w\|_2^2&=(1+\alpha^2)\|\x_{\Omega\setminus S}\|_2^2, \\
\label{e:alphauw}
\|\alpha^2\u-\w\|_2^2&=\alpha^2(1+\alpha^2)\|\x_{\Omega\setminus S}\|_2^2.
\end{align}
Thus,
\begin{align*}
\,&\w^T\B^T\B\u\\
\overset{(a)}{=}&\frac{\alpha\|\x_{\Omega\setminus S}\|_2}{\sqrt{N}}\overline{\e}_{W}^T\A^T_{W}
(\P^{\bot}_S)^T\P^{\bot}_S\A_{\Omega\setminus S}\x_{\Omega\setminus S}\\
\overset{(b)}{=}&\frac{\alpha\|\x_{\Omega\setminus S}\|_2}{\sqrt{N}}\overline{\e}_{W}^T\A^T_{W}
\P^{\bot}_S\A_{\Omega\setminus S}\x_{\Omega\setminus S}\\
\overset{(c)}{=}&\frac{\alpha\|\x_{\Omega\setminus S}\|_2}{\sqrt{N}}\sum_{j\in W}|\A_{j}^T\P^{\bot}_S\A_{\Omega\setminus S}\x_{\Omega\setminus S}|,
\end{align*}
where (a) follows from \eqref{e:B}-\eqref{e:AB}; (b) follows from \eqref{e:orthcom}, and (c) is from \eqref{e:sum}.
Therefore, we have
\begin{align}
&\|\B(\u+\w)\|_2^2-\|\B(\alpha^2\u-\w)\|_2^2\nonumber \\
=&(1-\alpha^4)\|\B\u\|_2^2+2(1+\alpha^2)\w^T\B^T\B\u\nonumber \\
=&(1-\alpha^4)\left(\|\B\u\|_2^2+\frac{2}{1-\alpha^2}\w^T\B^T\B\u\right)\nonumber \\
=&(1-\alpha^4)\big(\|\B\u\|_2^2\nonumber \\
&\quad+\frac{2\alpha}{1-\alpha^2}\frac{\|\x_{\Omega\setminus S}\|_2}{\sqrt{N}}
\sum_{j\in W}|\A_{j}^T\P^{\bot}_S\A_{\Omega\setminus S}\x_{\Omega\setminus S}|\big)\nonumber \\
=&(1-\alpha^4)\big(\|\B\u\|_2^2\nonumber \\
&\quad-\frac{\sqrt{|\Omega|-\ell}\|\x_{\Omega\setminus S}\|_2}{N}\sum_{j\in W}|\A_{j}^T\P^{\bot}_S
\A_{\Omega\setminus S}\x_{\Omega\setminus S}|\big),
\label{e:transf11}
\end{align}
where the last equality follows from the first equality in \eqref{e:alphaproperty}.

On the other hand, we have
\begin{align}
\label{e:transf12}
\,\;&\|\B(\u+\w)\|_2^2-\|\B(\alpha^2\u-\w)\|_2^2 \nonumber \\
\overset{(a)}{\geq}&(1-\delta_{N(k+1)+|\Omega|-\ell})\|(\u+\w)\|_2^2\nonumber \\
&-(1+\delta_{N(k+1)+|\Omega|-\ell})\|(\alpha^2\u-\w)\|_2^2\nonumber \\
\overset{(b)}{=}&(1-\delta_{N(k+1)+|\Omega|-\ell})(1+\alpha^2)\|\x_{\Omega\setminus S}\|_2^2\nonumber \\
&-(1+\delta_{N(k+1)+|\Omega|-\ell})\alpha^2(1+\alpha^2)\|\x_{\Omega\setminus S}\|_2^2\nonumber \\
=&(1+\alpha^2)\|\x_{\Omega\setminus S}\|_2^2\big[(1-\delta_{N(k+1)+|\Omega|-\ell})\nonumber \\
&-(1+\delta_{N(k+1)+|\Omega|-\ell})\alpha^2\big]\nonumber \\
=&(1+\alpha^2)\|\x_{\Omega\setminus S}\|_2^2\big[(1-\alpha^2)-\delta_{N(k+1)+|\Omega|-\ell}(1+\alpha^2)\big]\nonumber \\
=&(1-\alpha^4)\|\x_{\Omega\setminus S}\|_2^2\big(1-\frac{1+\alpha^2}{1-\alpha^2}\delta_{N(k+1)+|\Omega|-\ell}\big)\nonumber \\
\overset{(c)}{=}&(1-\alpha^4)\|\x_{\Omega\setminus S}\|_2^2\big(1-\sqrt{(|\Omega|-\ell)/N+1}\delta_{N(k+1)+|\Omega|-\ell}\big),
\end{align}
where (a) follows from \eqref{e:B} and Lemma \ref{l:orthogonalcomp} (note that $|\Omega\cap S|=\ell$, $|W|=N$ and $|S|=kN$, leading to
$|S\bigcup \left((\Omega\setminus S)\bigcup W\right)|=N(k+1)+|\Omega|-\ell$ ),
(b) follows from \eqref{e:uw} and \eqref{e:alphauw},
and (c) follows from the second equality in \eqref{e:alphaproperty}.

By \eqref{e:AB}, \eqref{e:transf11}, \eqref{e:transf12} and the fact that $1-\alpha^4>0$, we have
\begin{align*}
\,\;&\|\P^{\bot}_S\A_{\Omega\setminus S}\x_{\Omega\setminus S}\|_2^2\nonumber \\
&\quad-\frac{\sqrt{|\Omega|-\ell}\|\x_{\Omega\setminus S}\|_2}{N}\sum_{j\in W}|\A_{j}^T\P^{\bot}_S
\A_{\Omega\setminus S}\x_{\Omega\setminus S}|\nonumber \\
\geq&\|\x_{\Omega\setminus S}\|_2^2\big(1-\sqrt{(|\Omega|-\ell)/N+1}\delta_{N(k+1)+|\Omega|-\ell}\big).
\end{align*}
Thus, by \eqref{e:Ax}, we obtain
\begin{align*}
&\max_{i\in \Omega\setminus S}|\A_{i}^T\P^{\bot}_S\A_{\Omega\setminus S}\x_{\Omega\setminus S}|
-\frac{1}{N}\sum_{j\in W}|\A_{j}^T\P^{\bot}_S\A_{\Omega\setminus S}\x_{\Omega\setminus S}| \nonumber \\
\geq&\frac{\|\x_{\Omega\setminus S}\|_2\big(1-\sqrt{(|\Omega|-\ell)/N+1}\delta_{N(k+1)+|\Omega|-\ell}\big)}{\sqrt{|\Omega|-\ell} }
.
\end{align*}
Therefore, Lemma \ref{l:main} holds.
\ \ $\Box$

\section{Proof of Theorem \ref{t:kstep}}
\label{ss:kstep}



We prove the result by induction.
Suppose that gOMP selects at least one correct index in the first $k$ iterations,
then $\ell=|S_{k}\cap \Omega|\geq k$.
We assume $\Omega \not\subseteq S_{k}$ (i.e., $\ell\leq |\Omega|-1$) and Algorithm \ref{a:gOMP} performs at least $k+1$ iterations,
otherwise, the result holds.
Then, we need to show that $(S_{k+1}\setminus S_{k})\cap \Omega\neq \emptyset$.
Since $S_0=\emptyset$, the induction assumption $|\Omega|>|S_{k}\cap \Omega|\geq k$ holds with $k=0$.
Thus, the proof for the first iteration is contained in the case that $k=0$.

Let
 \beq
\label{e:W}
W= \{j_1,j_2,\ldots, j_N\} \subseteq\Omega^c
\eeq
such that
\begin{align}
\label{e:decreasing}
|\A_{j_1}^T\rr^{k}|&\geq\ldots\geq|\A_{j_N}^T\rr^{k}|\geq|\A_{j\in(\Omega^c\setminus W) }^T\rr^{k}|.
\end{align}
Then to show $(S_{k+1}\setminus S_{k})\cap \Omega\neq \emptyset$, we only need to show
\beqnn
\max_{i\in\Omega } |\A_i^T\rr^{k}|>|\A_{j_N}^T\rr^{k}|.
\eeqnn
By \eqref{e:decreasing},
\[
|\A_{j_N}^T\rr^{k}|\leq \frac{1}{N}\sum_{j\in W}|\A_{j}^T\rr^{k}|.
\]
Thus, to show $(S_{k+1}\setminus S_{k})\cap \Omega\neq \emptyset$, it suffices to show
\beq
\label{e:cond}
\max_{i\in\Omega } |\A_i^T\rr^{k}|>\frac{1}{N}\sum_{j\in W}|\A_{j}^T\rr^{k}|.
\eeq

By lines 4 and 5 of Algorithm \ref{a:gOMP}, we have
\begin{align}
\label{e:rk-1}
\rr^{k}&=\y-\A_{S_{k}}\hat{\x}_{S_{k}}
=\big(\I-\A_{S_{k}}(\A_{S_{k}}^T\A_{S_{k}})^{-1}\A_{S_{k}}^T\big)\y \nonumber \\
&\overset{(a)}{=}\P^{\perp}_{S_{k}}(\A\x+\v)
\overset{(b)}{=}\P^{\perp}_{S_{k}}(\A_{\Omega}\x_{\Omega}+\v) \nonumber \\
&=\P^{\bot}_{S_{k}}(\A_{\Omega\cap S_{k}}\x_{\Omega\cap S_{k}}+\A_{\Omega\setminus S_{k}}\x_{\Omega\setminus S_{k}}+\v)\nonumber\\
&\overset{(c)}{=}\P^{\perp}_{S_{k}}\A_{\Omega\setminus S_{k}}\x_{\Omega\setminus S_{k}}+\P^{\perp}_{S_{k}}\v,
\end{align}
where (a), (b) and (c) follow from the definition of $\P^{\perp}_{S_{k}}$,
the fact that $\Omega$ is the support of $\x$ and $\P^{\perp}_{S_{k}}\A_{S_{k}}=\0$, respectively.

By lines 3 and 4 of Algorithm \ref{a:gOMP}, for each $i\in S_{k}$,
$
|\A_i^T\rr^{k}|=0.
$
Thus, by \eqref{e:rk-1} and the triangular inequality, we have
\begin{align*}
\max_{i\in\Omega } |\A_i^T\rr^{k}|
\geq&
\max_{i \in \Omega\setminus S_{k}} (|\A_i^T\P^{\bot}_{S_{k}}\A_{\Omega\setminus S_{k}}\x_{\Omega\setminus S_{k}}|
-|\A_i^T\P^{\perp}_{S_{k}}\v|),\\
\frac{1}{N}\sum_{j\in W}|\A_{j}^T\rr^{k}|
\leq&
\frac{1}{N}\sum_{j\in W}|\A_j^T\P^{\bot}_{S_{k}}\A_{\Omega\setminus S_{k}}\x_{\Omega\setminus S_{k}}|\\
&+\max_{j\in W}|\A_j^T\P^{\perp}_{S_{k}}\v|.
\end{align*}
(Note that instead of lower bounding $\max_{i\in\Omega } |\A_i^T\rr^{k}|$ directly,
it was first lower bounded by $\frac{\|\A_{\Omega}^T\rr^{k}\|}{\sqrt {N}}$,
and then a lower bound on the latter quantity is given as a lower bound on $\max_{i\in\Omega } |\A_i^T\rr^{k}|$ in \cite[eq. (13)-(18)]{WanKS12}, this process requires $N\leq K$.)
Thus, to show \eqref{e:cond}, it suffices to show
\begin{align}
\label{e:condupdate}
\beta_1>\beta_2,
\end{align}
where
\begin{align}
\label{e:beta1}
\beta_1=&\max_{i \in \Omega\setminus S_{k}} |\A_i^T\P^{\bot}_{S_{k}}\A_{\Omega\setminus S_{k}}\x_{\Omega\setminus S_{k}}|\nonumber\\
&-\frac{1}{N}\sum_{j\in W}|\A_j^T\P^{\bot}_{S_{k}}\A_{\Omega\setminus S_{k}}\x_{\Omega\setminus S_{k}}|,\\
\label{e:beta2}
\beta_2=&\max_{i\in \Omega\setminus S_{k}} |\A_i^T\P^{\perp}_{S_{k}}\v|+\max_{j\in W} |\A_j^T\P^{\perp}_{S_{k}}\v|.
\end{align}

In the following, we apply the technique used in the proof of \cite[Theorem 1]{WenZWTM15} to give  an upper bound on $\beta_2$.
Clearly there exist $i_0\in\Omega\setminus S_{k}$ and $j_0\in W$ such that
\begin{align*}
\max_{i\in \Omega\setminus S_{k}} |\A_i^T\P^{\perp}_{S_{k}}\v|&=|\A^T_{i_0}\P^{\perp}_{S_{k}}\v|,\\
\max_{j\in W} |\A_j^T\P^{\perp}_{S_{k}}\v|&=|\A^T_{j_0}\P^{\perp}_{S_{k}}\v|.
\end{align*}
Therefore
\begin{align}
\label{e:beta2ub}
\beta_2=&\|\A^T_{i_0\cup j_0}\P^{\perp}_{S_{k}}\v\|_1
\overset{(a)}{\leq}\sqrt{2}\|\A^T_{i_0\cup j_0}\P^{\perp}_{S_{k}}\v\|_2\nonumber \\
\overset{(b)}{\leq}& \sqrt{2(1+\delta_{N(k+1)+|\Omega|-k})}\|\v\|_2,
\end{align}
where (a) is because $\A^T_{i_0\cup j_0}\P^{\perp}_{S_{k}}\v$ is a $2\times1$ vector,
(b) follows from Lemma \ref{l:AtRIP}, and
\beqnn
\|\P^{\perp}_{S_{k}}\v\|_2\leq\|\P^{\perp}_{S_{k}}\|_2\|\v\|_2\leq\|\v\|_2\leq\epsilon.
\eeqnn

In the following, we give  a lower bound on $\beta_1$.
By line 3 of Algorithm \ref{a:gOMP}, $|S_{k}|=kN$.
By the induction assumption,
\beq
\label{e:kl}
0\leq k\leq |\Omega\cap S_{k}|=\ell\leq |\Omega|-1.
\eeq
By \eqref{e:W}, $W\subset \Omega^c$ and $|W|=N$.
Thus, by Lemmas \ref{l:main} and \ref{l:monot}, and \eqref{e:beta1}, we obtain
\begin{align}
\label{e:betald1}
\beta_1\geq &\frac{(1-\sqrt {(|\Omega|-\ell)/N+1 }\delta_{N(k+1)+|\Omega|-\ell})\|\x_{\Omega\setminus S_{k}}\|_2}{\sqrt{|\Omega|-\ell}}\nonumber\\
\geq &\frac{(1-\sqrt {|\Omega|/N+1 }\delta_{N(k+1)+|\Omega|-k})\|\x_{\Omega\setminus S_{k}}\|_2}{\sqrt{|\Omega|-\ell}},
\end{align}
where the second inequality follows from \eqref{e:kl}, the fact that $k\leq k$ and Lemma \ref{l:monot}.

By \cite[eq.(21)]{Wan15}, we have
\beq
\label{e:betald2}
\|\x_{\Omega\setminus S_{k}}\|_2\geq\sqrt{\frac{|\Omega|-\ell}{K(1+\delta_{N(k+1)+|\Omega|-k})}}\sqrt{\mbox{MAR}\cdot\mbox{SNR}}\|\v\|_2.
\eeq
In fact, by the fact that $\ell=|\Omega\cap S_{k}|$, we have
\begin{align*}
\|\x_{\Omega\setminus S_{k}}\|_2&\geq\sqrt{|\Omega|-\ell}\min_{i\in\Omega }|x_i|\\
&\overset{(a)}{=}\sqrt{|\Omega|-\ell}\left(\sqrt{\mbox{MAR}}\|\x\|_2/\sqrt{K}\right)\\
&\overset{(b)}{\geq}\sqrt{\frac{|\Omega|-\ell}{K(1+\delta_{N(k+1)+|\Omega|-k})}}\sqrt{\mbox{MAR}}\|\A\x\|_2\\
&\overset{(c)}{\geq}\sqrt{\frac{|\Omega|-\ell}{K(1+\delta_{N(k+1)+|\Omega|-k})}}\sqrt{\mbox{MAR}\cdot\mbox{SNR}}\|\v\|_2,
\end{align*}
where (a) is from \eqref{e:SNR}, (b) is from
\begin{align*}
\|\A\x\|_2&=\|\A_{\Omega}\x_{\Omega}\|_2\leq \sqrt{1+\delta_{|\Omega|}}\|\x_{\Omega}\|_2\\
&\leq\sqrt{1+\delta_{N(k+1)+|\Omega|-k}}\|\x\|_2,
\end{align*}
and (c) follows from \eqref{e:SNR}.

By \eqref{e:betald1} and \eqref{e:betald2}, we have
\[
\beta_1\geq\frac{(1-\sqrt {|\Omega|/N+1 }\delta_{N(k+1)+|\Omega|-k})\sqrt{\mbox{MAR}\cdot\mbox{SNR}}\|\v\|_2}
{\sqrt{K(1+\delta_{N(k+1)+|\Omega|-k})}}.
\]
Thus, by \eqref{e:beta2ub}, \eqref{e:condupdate} can be guaranteed by
\begin{align*}
&\frac{(1-\sqrt {|\Omega|/N+1 }\delta_{N(k+1)+|\Omega|-k})\sqrt{\mbox{MAR}\cdot\mbox{SNR}}\|\v\|_2}
{\sqrt{K(1+\delta_{N(k+1)+|\Omega|-k})}}\\
>&\sqrt{2(1+\delta_{N(k+1)+|\Omega|-k})}\|\v\|_2,
\end{align*}
which is equivalent to \eqref{e:SNRld}.
By induction, the theorem holds.
\ \ $\Box$

\end{appendices}

\bibliographystyle{IEEEtran}
\bibliography{ref-RIP}

\end{document}